\documentclass[journal]{IEEEtran}
\usepackage{ifpdf}
\IEEEoverridecommandlockouts    
\usepackage[latin1]{inputenc}
\usepackage{amsmath}
\usepackage{subfigure}
\usepackage[pdftex]{graphicx}
\usepackage[usenames]{color}
\usepackage[colorlinks]{hyperref}
\usepackage{setspace}
\usepackage{comment}
\usepackage{tikz}
\usepackage{xspace}
\usepackage{color,soul}
\definecolor{yellow}{rgb}{.90,.95,1}
\sethlcolor{yellow}
\usepackage{amsopn}
\usepackage{amsopn}
\makeatletter
\let\NAT@parse\undefined
\makeatother
\usepackage{mathtools}
\usepackage{amsmath,amsfonts, amssymb,amsthm,color}
\usepackage[square,sort,comma,numbers]{natbib}
\usepackage{amsfonts}
\usepackage{color}

\usepackage{amssymb}
\newtheorem{theorem}{Theorem}
\newtheorem{problem}[theorem]{Problem}

\newtheorem{definition}{Definition}
\newtheorem{proposition}[definition]{Proposition}
\newtheorem{corollary}[definition]{Corollary}
\newtheorem{lemma}[theorem]{Lemma}
\newtheorem{remark}{Remark}

\usepackage{breqn}
\usepackage{hyperref} 
\usepackage{amsxtra}     
\usepackage{amsthm}
\usepackage{epstopdf} 
\usepackage{epsfig}
\DeclareGraphicsExtensions{.pdf,.eps,.png,.jpg,.mps} 
\usepackage[]{algorithmicx}
  \usepackage[ruled]{algorithm}
\usepackage{algpseudocode}
\usepackage[utf8]{luainputenc}

\title{ \huge Distributed Resource Allocation for Epidemic control}
\author{Chinwendu Enyioha, Ali Jadbabaie, Victor Preciado and George Pappas
\thanks{*This research was supported by ARO MURI W911NF-12-1-0509  and AFOSR Complex Networks Program}
\thanks{All authors are with the Department of Electrical and Systems Engineering, University of Pennsylvania, Philadelphia PA, USA 19104
        {\tt\small \{cenyioha, jadbabai, preciado, pappasg\} @seas.upenn.edu}
        }%
}
\begin{document}

\maketitle

\begin{abstract}

We present a distributed resource allocation strategy to control an epidemic outbreak in a networked population based on a Distributed Alternating Direction Method of Multipliers (D-ADMM) algorithm. We consider a linearized Susceptible-Infected-Susceptible (SIS) epidemic spreading model in which agents in the network are able to allocate vaccination resources (for prevention) and antidotes (for treatment) in the presence of a contagion. We express our epidemic control condition as a spectral constraint involving the Perron-Frobenius eigenvalue, and formulate the resource allocation problem as a Geometric Program (GP). Next, we separate the network-wide optimization problem into subproblems optimally solved by each agent in a fully distributed way. We conclude the paper by illustrating performance of our solution framework with numerical simulations.
\end{abstract}
\section{Introduction}
In this paper, we consider the problem of controlling the outbreak of an epidemic process in a networked population in the absence of a social planner. We propose a distributed solution that enables local computation of optimal investment in vaccines and/or antidotes at each node, to respectively reduce its infection rate and increase its recovery rate; and contain the spread of an outbreak. Our distributed solution is based on an iterative algorithm -- a \textit{Distributed Alternating Direction Method of Multipliers} (D-ADMM) algorithm.
The distributed solution framework presented here is based on a Geometric Programming formulation of the epidemic control problem, and its convex characterization first proposed in \cite{preciado2013optimalgp}.

The fairly recent outbreaks of the Middle Eastern Respiratory Syndrome (MERS) virus \cite{de2013middle} and Ebola virus \cite{du2014ebola} have rekindled interest in analysis and control of epidemic outbreaks in networked populations. The analysis of an SIS epidemic process in arbitrary (undirected) networks was first studied by Wang et al. \cite{WCWF03} using a discrete-time model. Key in \cite{WCWF03} was the establishment of an epidemic threshold based on the spectral radius of the network adjacency matrix. Ganesh et al. in \cite{GMT05} studied a continuous-time SIS epidemic model where the relation between the spectral radius of the network adjacency matrix and the speed of spreading was further corroborated. A similar result was proposed in \cite{van2009virus}.

In the event of an epidemic outbreak, of interest is how to optimally allocate vaccines and/or antidotes across the the population to control the spread of the epidemic. Sometimes, the focus is to hasten the rate at which the contagion is contained; or to minimize the associated cost of containing the outbreak. Significant research attention has been given to this problem in the last decade. In \cite{cohen2003efficient}, the authors proposed a greedy immunization strategy -- immunizing acquaintances, and showed it outperforms a case where vaccinations are applied randomly across nodes in the network. In \cite{chung2009distributing} the authors presented an analysis of immunization strategy using the PageRank vector of the network adjacency matrix. Wan et al. in \cite{WRS08} presented a method to design control strategies in undirected networks using eigenvalue sensitivity analysis. A convex framework to compute the optimal allocation of vaccines in undirected networks using Semidefinite programming (SDP) techniques was discussed in \cite{PZEJP13}. In \cite{preciado2013optimalgp}, we generalized the convex formulation via Geometric Programming techniques to weighted, directed, strongly connected networks to compute the cost and speed optimal allocation of vaccines and/or antidotes in directed strongly (and not necessarily strongly) connected networks.

In the absence of a social planner (or central coordinator), centralized computation of optimal investments in vaccines and/or antidotes as was done in \cite{preciado2013optimalgp} poses a challenge. Furthermore, the network might be too large for a centralized optimization scheme. These motivate the need for an optimal, decentralized framework to achieve the network-wide objective of controlling an outbreak. In this paper, we propose a fully distributed ADMM algorithm that enables local computation of optimal investment in vaccines and/or antidotes at each node to control the spread of an outbreak in a directed, network comprising heterogeneous agents.  While the discussion in this paper focuses on strongly connected networks, the results are easily amenable to networks that are not necessarily strongly connected.

 We formulate the epidemic control problem as a distributed resource allocation problem, in that agents in the network carry out local computations informed by interactions with neighboring agents to determine their optimal investment in vaccines and/or antidotes, and must cooperate to complete a global task.  Our work is similar in spirit to  \cite{rantzer2011distributed} where tests for distributed control of positive systems were presented. A notable difference and major challenge between the particular problems addressed in \cite{rantzer2011distributed} and this paper is that our problem in its natural form is a nonconvex optimization problem. A game-theoretic formulation has also been studied in \cite{trajanovski2014decentralized} where different equilibria were analyzed for undirected networks, contrasting our distributed solution that applies to directed networks with non-identical agents.

Distributed optimization methods exist in the literature and have applied to several classes of problems; see, for example \cite{jakovetic2011cooperative,nedic2009distributed, tsitsiklis1984problems}. 
We employ a D-ADMM, a dual-based method, to solve our constrained resource allocation problem. Dual-based methods are usually used when the local optimization of each agent can be done efficiently. The D-ADMM, as we will illustrate, first decomposes our original resource allocation problem into two sub-problems. These sub-problems are then sequentially solved in parallel by each agent, with the associated dual variables updated at each iteration of the algorithm, allowing for a fully distributed implementation. In contrast to \cite{wei2012distributed}, where a D-ADMM algorithm was proposed for an unconstrained optimization problem on an undirected network, our problem in its original form is a  constrained optimization problem on a directed network, where the need for satisfaction of a global constraint is necessary. 

The rest of the paper is organized as follows -- In section \ref{sec:model}, we present the notation, spreading model considered, as well as state the resource allocation problem. Section \ref{sec:convex_character} comprises a  GP formulation of the problem, alongside its convex characterization. The proposed distributed solution is presented in Section \ref{sec:distributed}; with results from our solution illustrated in Section \ref{sec:experiments}. A summary and concluding remarks follow in Section \ref{sec:summary}.

\section{Preliminaries, Model and Problem Statement}
\label{sec:model}
\subsection{Graph Theory}
 
We define a weighted graph as  $\mathcal{G}\triangleq\left(\mathcal{V},\mathcal{E},\mathcal{W}\right)$, where $\mathcal{V}\triangleq\left\{ v_{1},\dots,v_{n}\right\} $
is a set of $n$ nodes, $\mathcal{E}\subseteq\mathcal{V}\times\mathcal{V}$
is a set of ordered pairs of nodes called edges, and the function $\mathcal{W}:\mathcal{E}\rightarrow\mathbb{R}_{++}$
associates \textit{positive} real weights to the edges in $\mathcal{E}$.
The node pair $\left(v_{j},v_{i}\right)$ form an undirected edge in $\mathcal{G}$. If $\mathcal{G}$ is a directed network, the pair $\left(v_{j},v_{i}\right)$ is an oriented edge from node $v_j$ to node $v_i$.  For an undirected graph $\mathcal{G}$, we define the neighborhood of node $v_i$ as $N_{i} \triangleq\left\{ j:\left(v_{j},v_{i}\right)\in\mathcal{E}\right\} $. When the graph $\mathcal{G}$, is directed, we define the in-neighborhood of node $v_i$ (that is, the set of nodes with edges pointing towards $v_{i}$), as $N_{i}^{in}\triangleq\left\{ j:\left(v_{j},v_{i}\right)\in\mathcal{E}\right\} $. 
A directed path from $v_{i_{1}}$ to $v_{i_{l}}$ in $\mathcal{G}$ is an ordered set of vertices $\left(v_{i_{1}},v_{i_{2}},\ldots,v_{i_{l+1}}\right)$ such that $\left(v_{i_{s}},v_{i_{s+1}}\right)\in\mathcal{E}$ for $s=1,\ldots,l$. A directed graph $\mathcal{G}$ is \emph{strongly connected} if, for every pair of nodes $v_{i},v_{j}\in\mathcal{V}$, there is a directed path from $v_{i}$ to $v_{j}$.

We denote the adjacency matrix of a directed graph $\mathcal{G}$  by $A = [a_{ij}] = \mathcal{W}(v_j,v_i)$ if the edge $(v_j,v_i)\in \mathcal{E}$ points from $v_j$ to $v_i$, and $a_{ij} = 0$ otherwise. Given an $n\times n$ matrix $M$, we denote
by $\mathbf{v}_{1}\left(M\right),\ldots,\mathbf{v}_{n}\left(M\right)$
and $\lambda_{1}\left(M\right),\ldots,\lambda_{n}\left(M\right)$
the set of eigenvectors and corresponding eigenvalues of $M$, respectively,
where we order them in decreasing order of their real parts, i.e.,
$\mathbb{R}(\lambda_{1}) \geq \mathbb{R}(\lambda_2)\geq\ldots\geq \mathbb{R}(\lambda_n)$.
We respectively call $\lambda_{1}\left(M\right)$ and $\mathbf{v}_{1}\left(M\right)$
the dominant eigenvalue and eigenvector of $M$; and denote by  $\rho\left(M\right)$, the spectral radius of $M$, which is the maximum modulus across all eigenvalue of $M$. 
Vectors are denoted using boldface letters
and matrices using capital letters. The letter $I$ denotes the identity matrix and $\mathbb{R}(z)$ denotes
the real part of $z\in\mathbb{C}$. 

\subsection{Spreading Model}
The spreading model considered in this paper is a continuous-time heterogeneous (SIS) epidemic model and closely follows the development in \cite{enyiohathesis}. In the heterogeneous model, we assume a networked Markov process in which nodes in the network can be in one of two states -- susceptible or infected. Over time, based on each agent's  infection rate $\beta_i$, recovery rate $\delta_i$ and interactions with neighboring agents and their states, the (infected or susceptible) state of each agent evolves. We assume the agents are non-identical in the sense that each agent has distinct infection rate $\beta_i$ and recovery rate $\delta_i$; and that these rates can be controlled by either injecting vaccines or antidotes.

The evolution of the spreading model is as follows. At each time-step, the state of node $v_i$ is a binary random variable $X_i(t) \in \{0,1\}$. The states $X_i(t) = 0$ and $X_i(t) = 1$ respectively  indicate that node $v_i$ is susceptible and infected. Let the vector of states be $X(t) = (X_1(t),\hdots, X_n(t))^T$. Suppose node $v_i$ is in the susceptible state at time $t$. Its probability of switching to the infected state depends on its infection rate $\beta_i$, the state of its neighbors $\{X_j(t), \ \text{for } j\in \mathcal{N}_i^{in}\}$. Mathematically, the probability of node $v_i$ switching from susceptible to infected state can be expressed as
\begin{multline}
\Pr\left(X_{i}(t+\Delta t)=1|X_{i}(t)=0,X(t)\right)=\\
\sum_{j\in\mathcal{N}_{i}^{in}}a_{ij}\beta_{i}X_{j}\left(t\right)\Delta t+o(\Delta t),
\end{multline}
where $\Delta t>0$ is considered an asymptotically small time interval. Similarly, suppose node $v_i$ is in the infected; the probability of node 
$v_i$ returning to the susceptible state depends on its recovery rate $\delta_i$, as well as the states of its neighbors with incoming connections. More formally, this is given by
\begin{equation}
\Pr(X_{i}(t+\Delta t)=0|X_{i}(t)=1,X(t))=\delta_{i}\Delta t+o(\Delta t).
\end{equation}
Given an $n-$agent network and two possible states each agent can be in, the Markov process has $2^n$ states. The exponentially increasing state space of this model poses computational challenges and makes this model difficult to analyze for large networks. To overcome this challenge, we
use a mean-field approximation of its dynamics \cite{van2009virus}. By applying mean-field theory to the Markov chain, the dynamics
of infection spread can be approximated using a system of
differential equations. Suppose we define $p_{i}\left(t\right)\triangleq\Pr\left(X_{i}\left(t\right)=1\right)=\mathbb{E}\left(X_{i}\left(t\right)\right)$,
i.e., the marginal probability of node $v_{i}$ being infected at
time $t$. Then, the Markov differential equation \cite{van2009virus}
for the state $X_{i}\left(t\right)=1$ is:
\begin{equation}
\frac{dp_{i}\left(t\right)}{dt}=\left(1-p_{i}\left(t\right)\right)\beta_{i}\sum_{j=1}^{n}a_{ij}p_{j}\left(t\right)-\delta_{i}p_{i}\left(t\right).\label{eq:HeNiSIS dynamics}
\end{equation}
Since $i=1,\ldots,n$, we can represent \eqref{eq:HeNiSIS dynamics} more compactly by a system of nonlinear differential
equation with dynamics
\begin{equation}\label{eq:dynamics}
\frac{d\mathbf{p}(t)}{dt} = (BA - D)\mathbf{p}(t) - P(t)BA\mathbf{p}(t),
\end{equation}
where $\mathbf{p}\left(t\right)\triangleq\left(p_{1}\left(t\right),\ldots,p_{n}\left(t\right)\right)^{T}$,
$B\triangleq\mbox{diag}(\beta_{i})$, $D\triangleq\mbox{diag}\left(\delta_{i}\right)$,
and $P\left(t\right)\triangleq\mbox{diag}(p_{i}\left(t\right))$.
This ODE presents an equilibrium point at $\mathbf{p}^{*}=0$,
called the disease-free equilibrium. Readers are referred to \cite{van2009virus, bose2013cost} for a detailed treatment and presentation of the mean-field approximation noted above.
\begin{proposition}\label{prop:stab}
Consider the nonlinear dynamical system in \eqref{eq:dynamics}, where the network adjacency matrix $A \geq 0$, $B = \text{diag}(\beta_i) \geq 0$ and $D = \text{diag}(\delta_i) > 0$. Suppose the real part of the largest eigenvalue of $BA - D$ satisfies 
\begin{equation}\label{stability_constraint}
\mathbb{R}\left(\lambda_1(BA - D)  \right) \leq -\varepsilon,
\end{equation}
for $\varepsilon > 0$, the disease-free equilibrium is globally exponentially stable; that is, $\|\mathbf{p}(t) \| \leq \|\mathbf{p}(0)\|  K e^{-\varepsilon t}$, for $K >0$.
\end{proposition}
\begin{proof}
See \cite{PZEJP13,enyiohathesis} for proof; where we showed that the linear dynamical system $\dot{\mathbf{p}}\left(t\right)=\left(BA -D\right)\mathbf{p}\left(t\right)$
upper-bounds the dynamics in (\ref{eq:HeNiSIS dynamics});
thus, the spectral result in \eqref{stability_constraint} is a sufficient
condition for the mean-field approximation of \eqref{eq:dynamics} to be globally exponentially stable.
\end{proof}

\subsection{The Distributed Resource Allocation Problem}
We present the epidemic control problem as a resource allocation problem in which vaccines and antidotes are respectively applied at each node to reduce their infection rates and increase their recovery rates within feasible intervals $0 < \underline{\beta}_i \leq \beta_i \leq \overline{\beta}_i$ and  $0 < \underline{\delta}_i \leq \delta_i \leq \overline{\delta}_i$. The specific values of $\beta_i$ and $\delta_i$ depend on the amount of vaccines and antidotes allocated at node $v_i$. We assume that vaccines applied at node $v_i$ have an associated cost $f_i(\beta_i)$; and antidotes have a cost $g_i(\delta_i)$. The function $f_i(\beta_i)$, assumed to be convex and monotonically decreasing with respect to $\beta_i \ \ \forall \ i=1,\hdots, n$, represents the cost of modifying the infection rates of node $v_i$ to some $\beta_i$ within the feasible interval to stabilize the spreading dynamics. Similarly, the function $g_i(\delta_i)$ is monotonically nondecreasing with respect to $\delta_i \ \ \forall \ i=1,\hdots, n$, and represents the cost of modifying the recovery rates of node $v_i$ to some $\delta_i$ within the feasible interval to contain the spread of an epidemic outbreak. 

We focus on a \textit{rate-constrained} resource allocation problem, in which the objective is to find the cost-optimal, local allocation of vaccines and antidotes to achieve a given network-wide exponential decay rate in the probability of infection. That is, given a desired infection decay rate $\overline{\varepsilon}$, the objective is to \textit{locally} allocate resources to each node $v_i$ such that $\left\Vert \mathbf{p}\left(t\right)\right\Vert \leq\left\Vert \mathbf{p}\left(0\right)\right\Vert K\exp\left(-\overline{\varepsilon}t\right)$,
$K>0$; where $\mathbf{p} = (p_1, \hdots, p_n)^T$, via local computations and interactions with neighbors. 
The problem is formally stated below:

\begin{problem}
\label{Problem:decentralized Rate Constrained Allocation}
Given a directed network with associated adjacency matrix $A$, node cost functions $\{f_i(\beta_i), g_i(\delta_i)  \}_{i=1}^n$ and bounds on the infection and curing rates $ 0<\underline{\beta}_{i}\leq\beta_{i}\leq\bar{\beta}_{i}$, and $0<\underline{\delta}_{i}\leq\delta_{i}\leq\bar{\delta}_{i}$ respectively, and an exponential decay rate $\bar{\varepsilon} > 0$. Locally (via interaction between neighboring nodes), determine the cost-optimal distribution of vaccines and treatment resources to attain the desired decay rate.
\end{problem}
Mathematically, the objective is to:
\begin{align}
\underset{\left\{ \beta_{i},\delta_{i}\right\} _{i=1}^{n}}{\mbox{minimize }} & \sum_{i=1}^{n}f_{i}\left(\beta_{i}\right)+g_{i}\left(\delta_{i}\right)\label{eq:Rate-Constrain Spectral Problem}\\
\mbox{subject to } & \mathbb{R}(\lambda_{1}\left(B
A - D \right))\leq-\bar{\varepsilon},\label{eq:Spectral constraint}\\
 & \underline{\beta}_{i}\leq\beta_{i}\leq\overline{\beta}_{i},\label{eq:Square constraint beta}\\
 & \underline{\delta}_{i}\leq\delta_{i}\leq\overline{\delta}_{i},\mbox{ }i=1,\ldots,n,\label{eq:Square constraint delta}
\end{align}
where (\ref{eq:Rate-Constrain Spectral Problem}) is the total investment across all agents,
(\ref{eq:Spectral constraint}) constrains the decay rate to $\bar{\varepsilon}$,
and (\ref{eq:Square constraint beta})-(\ref{eq:Square constraint delta})
maintain the infection and recovery rates in their feasible limits.
Our goal is to solve \eqref{eq:Rate-Constrain Spectral Problem}-\eqref{eq:Square constraint delta} in a fully distributed manner, with the computation of investment per node $f_i(\beta_i) + g_i(\delta_i)$, carried out locally.
After presenting a convex characterization of the problem, we illustrate how it decomposes nicely for a distributed solution.

\section{The Resource Allocation Problem as a GP}
\label{sec:convex_character}
In this section, we formulate the resource allocation problem in \eqref{eq:Rate-Constrain Spectral Problem}-\eqref{eq:Square constraint delta} as a Geometric Program and present its convex characterization.
Geometric Programs (GPs) are a class of nonlinear, nonconvex optimization problems that can be transformed into convex optimization problems and efficiently solved using interior-point methods, yielding globally optimal solutions  \cite{boyd2007tutorial}. Building blocks of a GP are monomial and posynomial functions. Let the vector $\mathbf{x}\triangleq (x_1, \hdots, x_n) \in \mathbb{R}^n_{++}$ denote n decision variables. A monomial functions (of the variables $x_1,\hdots, x_n$) is a real-valued function $f(\mathbf{x}) = c x_1^{a_1}  x_2^{a_2} \hdots x_n^{a_n},$ where $c>0$ is the coefficient of the monomial and $a_i \in \mathbb{R}$ are exponents of the monomial. Posynomial functions are sums of monomials; that is, a function of the form $f(\mathbf{x})  = \sum_{k=1}^K c_k x_1^{a_{1k}}  x_2^{a_{2k}} \hdots x_n^{a_{nk}},$ where $c_k > 0$, for $k=1, 2, \hdots, K$ and $a_{jk} \in \mathbb{R}$ for $j = 1, \hdots, n$ and $k = 1, \hdots, K$. For example, while $x_1+x_2^4$ and $x_1x_2^{0.3}\pi$ are posynomials, $x_1 - x_2^2$ is not a posynomial.  \quad A GP in standard form is one in which a posynomial function is minimized subject to posynomial upper bound inequality constraints and monomial equality constraints. More formally, it is of the form
\begin{equation}
\begin{aligned}
& \underset{}{\text{minimize}}
& & f_0(\mathbf{x}) \\
& \text{subject to}
& & f_i(\mathbf{x}) \leq 1, \ \ i=1, \hdots, m, \label{gpstand}\\
& & & g_i(\mathbf{x}) = 1, \ \ i=1, \hdots, p, \\
\end{aligned}
\end{equation}
where $f_i$ are posynomial functions and $g_i$ are monomials. GPs in standard form are not convex optimization problems, since posynomials are not convex functions. However, with a logarithmic change of variables and multiplicative constants: $y_i = \log x_i, \ \ b_l = \log c_l, \ \ b_{ik} = \log c_{ik}$ and a logarithmic change of the functions' values, we can transform \eqref{gpstand} to the following equivalent problem in the variable $\mathbf{y}$:
\begin{equation}
\begin{aligned}
& \underset{}{\text{minimize}}
& & h_0(\mathbf{y})  = \log \sum_{k=1}^{K_0} \exp(\mathbf{a}_{0k}^T\mathbf{y}+b_{0k}) \\
& \text{subject to}
& & h_i(\mathbf{y}) = \log \sum_{k=1}^{K_i} \exp(\mathbf{a}_{ik}^T\mathbf{y}+b_{ik})\leq 0, \ \ \forall \ i  \label{gpconvex}\\
& & & q_i(\mathbf{y}) = \mathbf{a}_l^T\mathbf{y} + b_l = 0, \ \ l=1, \hdots, M. \\
\end{aligned}
\end{equation}
Problem \eqref{gpconvex} is convex and can be efficiently solved in polynomial time. (See \cite{boyd2007tutorial} for a more detailed exposition on GPs). Hence, if our cost function is a posynomial (or more generally, convex in log scale), and if the constraints can be represented as monomial and posynomial functions, \eqref{eq:Rate-Constrain Spectral Problem}-\eqref{eq:Square constraint delta} can be formulated as a GP in convex form.

Via the Perron-Frobenius lemma, from the theory of nonnegative matrices, we express the spectral constraint \eqref{stability_constraint} in an equivalent form to transform it to a set of posynomial constraints in the decision variables.

\begin{lemma}
\label{lem:Perron-Frobenius}(Perron-Frobenius) Let $M$ be a nonnegative,
irreducible matrix. Then, the following statements about its spectral
radius, $\rho\left(M\right)$, hold:
\begin{enumerate}
\item $\rho\left(M\right)>0$ is a simple eigenvalue of
$M$,

\item $M\mathbf{u}=\rho\left(M\right)\mathbf{u}$, for
some $\mathbf{u}\in\mathbb{R}_{++}^{n}$, and

\item $\rho\left(M\right)=\inf\left\{ \lambda\in\mathbb{R}:M\mathbf{u}\leq\lambda\mathbf{u}\mbox{ for }\mathbf{u}\in\mathbb{R}_{++}^{n}\right\} $. \label{perron-equiv}
\end{enumerate}
\end{lemma}

Recall that our model assumes the contact network is directed, strongly connected. Since the adjacency matrix associated with a strongly connected, directed graph is irreducible, Lemma \ref{lem:Perron-Frobenius} holds for the spectral radius of the adjacency matrix of any positively weighted, strongly connected digraph. A corollary of Lemma \ref{lem:Perron-Frobenius} is the following:
\begin{corollary}
\label{cor:Eig equals Rad}Let $M$ be a nonnegative, irreducible
matrix. Then, its eigenvalue with the largest real part, $\lambda_{1}\left(M\right)$,
is real, simple, and equal to the spectral radius $\rho\left(M\right)>0$.
\end{corollary}

\noindent In \cite{preciado2013optimalgp}, based on Propositions \ref{prop:stab} and Lemma \ref{lem:Perron-Frobenius}, the following result established the formulation of \eqref{eq:Rate-Constrain Spectral Problem}-\eqref{eq:Square constraint delta} as a GP comprising monomial and posynomial functions.
\begin{theorem}(\cite{preciado2013optimalgp})
\label{thm:GP for rate constrained} Given a strongly connected graph $\mathcal{G}$ with adjacency matrix
$A=[A_{ij}]$, posynomial cost functions $\left\{ f_{i}\left(\beta_{i}\right),g_{i}\left(\delta_{i}\right)\right\} _{i=1}^{n}$, bounds on the infection and recovery rates $0<\underline{\beta}_{i}\leq\beta_{i}\leq\overline{\beta}_{i}$
and \textup{$0<\underline{\delta}_{i}\leq\delta_{i}\leq\overline{\delta}_{i}$},
$i=1,\ldots,n$, and a desired exponential decay rate $\overline{\varepsilon}$.
Then, the optimal investment on vaccines and antidotes for node $v_{i}$ to solve Problem \eqref{eq:Rate-Constrain Spectral Problem}-\eqref{eq:Square constraint delta}
are $f_{i}\left(\beta_{i}^{*}\right)$ and $g_{i}\left(\widetilde{\Delta}+1-\widetilde{\delta}_{i}^{*}\right)$,
where $\widetilde{\Delta}\triangleq\max\left\{ \overline{\varepsilon},\overline{\delta}_{i}\mbox{ for }i=1,\ldots,n\right\} $
and \textup{\emph{$\beta_{i}^{*}$,}}\textup{$\widetilde{\delta}_{i}^{*}$}\textup{\emph{
are the optimal solution for $\beta_{i}$ and $\widetilde{\delta}_{i}$
in the following GP}}:
\begin{align}
\underset{ \left\{ u_{i},\beta_{i},\widetilde{\delta}_{i},t_{i}\right\} _{i=1}^{n}}{\mbox{minimize}} & \sum_{k=1}^{n}f_{k}\left(\beta_{k}\right)+g_{k}\left(t_{k}\right)\label{eq:Budget-Constrained Spectral Problem-1-1}\\
\mbox{subject to } & \frac{\beta_{i}\sum_{j=1}^{n}A_{ij}u_{j}+\widetilde{\delta}_{i}u_{i}}{\left(\widetilde{\Delta}+1-\overline{\varepsilon}\right)u_{i}}\leq1,\label{eq:Spectral constraint in budget constrained}\\
 & \left(t_{i}+\widetilde{\delta}_{i}\right)\left/\left(\widetilde{\Delta}+1\right)\right.\leq1,\label{eq:t trick in rate constrained}\\
 & \widetilde{\Delta}+1-\overline{\delta}_{i}\leq\widehat{\delta}_{i}\leq\widetilde{\Delta}+1-\underline{\delta}_{i},\label{eq:delta constraint in rate constrained}\\
 & \underline{\beta}_{i}\leq\beta_{i}\leq\overline{\beta}_{i},\, i=1,\ldots,n.\label{eq:beta constraint in rate constrained}
\end{align}
\end{theorem}
\begin{proof}
See Theorem $12$ in \cite{preciado2013optimalgp} for proof.
\end{proof}
For simplicity of notation and ease of reading, we re-express \eqref{eq:Budget-Constrained Spectral Problem-1-1}-\eqref{eq:beta constraint in rate constrained} as 
\begin{align}
\underset{ u_i,\beta_i, \delta_i,}{\mbox{minimize}} &  \quad \sum_{i=1}^n f_i(\beta_i) + g_i(\delta_i) \label{eq:refined problem} \\
\mbox{subject to } & \frac{\beta_{i}\sum_{j=1}^{n}A_{ij}u_{j} +  \delta_i u_i}{u_i}\leq 1, \\
 & 
 \underline{\delta} \leq \delta_i \leq \overline{\delta} \label{eq:new-delta}\\
& 
\underline{\beta} \leq \beta_i \leq \overline{\beta}, \quad i=1,\hdots, n; \label{eq:new-beta}
\end{align}
where the auxiliary variables $\widetilde{\Delta}, t_i, \widetilde{\delta}_i$ introduced in Theorem \ref{thm:GP for rate constrained} to express the spectral constraint as a set of posynomial functions have been factored into the upper and lower bounds on $\delta_i$ in \eqref{eq:new-delta}. The rest of the paper will focus on developing a distributed solution to  \eqref{eq:refined problem}-\eqref{eq:new-beta}.

\subsection{Separability of \eqref{eq:refined problem}-\eqref{eq:new-beta}}
To implement a distributed solution, first note that the cost function and constraint functions of the GP in \eqref{eq:refined problem}-\eqref{eq:new-beta} is separable per agent $v_i$. Each agent is able to locally solve the following problem:
\begin{align}
\underset{ u_i,\beta_i, \delta_i,}{\mbox{minimize}} &  \quad  f_i(\beta_i) + g_i(\delta_i) \label{eq:separated problem} \\
\mbox{subject to } & \frac{\beta_{i}\sum_{j=1}^{n}A_{ij}u_{j} +  \delta_i u_i}{u_i}\leq 1, \label{eq:coupling}\\
 & 
 \underline{\delta} \leq \delta_i \leq \overline{\delta} \\
& 
\underline{\beta} \leq \beta_i \leq \overline{\beta}.\label{eq:separated problem4}
\end{align}
Though separable per agent, not all the decision variables are local. In particular, for agent $v_i$ to minimize its cost function it needs the value $u_j$ from nodes in its neighborhood set. The need for $u_j$ in computing the optimum cost of node $v_i$ is explicit in \eqref{eq:coupling}. To address this problem, we employ the Alternating Direction Method of Multipliers algorithm, which is well suited for such distributed optimization problems. 

\section{Distributed solution}
\label{sec:distributed}

\subsection{Alternating Direction Method of Multipliers (ADMM)}
The ADMM algorithm is a dual-based method for solving constrained optimization problems in which an augmented Lagrangian function is minimized with respect to the primal variables, and the dual variables are updated accordingly. Recent surveys on augmented Lagrangian methods and the ADMM algorithm can be found in monographs by Schizas et al., \cite{schizas2008consensus}, Bertsekas \cite{bertsekas1982constrained} and a survey by Boyd et al. \cite{boyd2011distributed}; where illustrations and solutions to different optimization problems via the ADMM algorithm are presented. The ADMM algorithm works by decomposing the original optimization problem into subproblems that can be sequentially solved in parallel by each agent, allowing for distributed solutions to large-scale optimization problems. 


The standard ADMM solves the following problem
\begin{equation}
\begin{aligned}
& \underset{x,z}{\text{minimize}}
& & f(x) + g(z) \\
& \text{subject to}
& & Ax + Dz = c, \label{admm:standard}
\end{aligned}
\end{equation}
where the variables $x\in \mathbb{R}^n, \ z\in \mathbb{R}^m$, matrices $A\in \mathbb{R}^{p \times n}, \ D\in \mathbb{R}^{p\times m}$, and $c\in \mathbb{R}^p$. The Augmented Lagrangian function for \eqref{admm:standard} is given by
\begin{equation}
L_{\rho}(x,z,\lambda) = f(x) + g(z) - \lambda^T(Ax + Dz - c) + \frac{\rho}{2} \|Ax + Dz - c \|^2,
\end{equation}
where $\lambda$ is the Lagrange multiplier associated with the constraint $Ax + Dz = c$ and $\rho$ is a positive scalar. The update rules for the variables $x, \ z$ and $\lambda$ in the ADMM implementation is given by 
\begin{align}
x(k+1)  &= \text{arg}\min_x L_{\rho}(x,z(k),\lambda(k))  \label{eq:updates} \\
z(k+1) &= \text{arg}\min_z  L_{\rho}(x(k+1),z,\lambda(k))  \label{eq:updates1}\\
\lambda(k+1) &= \lambda(k) - \rho\left((Ax(k+1) + Dz(k+1) - c\right) \label{eq:updates2}
\end{align}
The updates in \eqref{eq:updates} - \eqref{eq:updates2} is similar to those of dual descent algorithms (see \cite{bertsekas1989parallel} for instance), except that augmented Lagrangian is used and penalty parameter $\rho$ is used as the step size in the dual updates. 


\subsection{Resource Allocation via D-ADMM}
Our goal is to present a distributed solution to the GP in \eqref{eq:separated problem}-\eqref{eq:separated problem4}. To make \eqref{eq:separated problem}-\eqref{eq:separated problem4} amenable to a distributed solution via the ADMM, we introduce variables $\mathbf{u}_i \in \mathbb{R}^n$ representing a local copy of the global variable $\mathbf{u} = (u_1, \hdots, u_n)^T$ in \eqref{eq:coupling} at each node $v_i$. We also introduce an auxiliary variable $\mathbf{z}_{ij}$, that enables communication and enforces consensus in the values of $\mathbf{u}_i$ and $\mathbf{u}_j$ for all neighboring nodes $(v_i,v_j)\in \mathcal{E}$. We interpret the auxiliary variables $\mathbf{z}_{ij}$ as being associated with the edge $(v_i, v_j)$ with the goal of enforcing consensus of the variables $\mathbf{u}_i$ and $\mathbf{u}_j$ of its adjacent nodes $v_i$ and $v_j$. With these new variables, and consensus constraint, we can reformulate \eqref{eq:separated problem}-\eqref{eq:separated problem4} as
\begin{align}
\underset{ \mathbf{u}_i,\beta_i, \delta_i,}{\mbox{minimize}} &  \quad  f_i(\beta_i) + g_i(\delta_i) \nonumber \\
\mbox{subject to } & \frac{\beta_{i}\sum_{j=1}^{n}A_{ij}\mathbf{u}_i^j +  \delta_i \mathbf{u}^i_i}{\mathbf{u}^i_i}\leq 1, \label{eq:separated-problem1} \\
&  \prod_{j=1}^n \mathbf{u}_i^j = 1, \nonumber \\
 & 
 \underline{\delta} \leq \delta_i \leq \overline{\delta} \nonumber \\
& 
\underline{\beta} \leq \beta_i \leq \overline{\beta}. \nonumber \\
 & \mathbf{u}_i = \mathbf{z}_{ij} \quad and \quad \mathbf{u}_j = \mathbf{z}_{ij}, \ \ (v_i,v_j) \in \mathcal{E}, \nonumber
\end{align}
where the scalar $\mathbf{u}_i^j$ is the $j$th entry of the local estimate $\mathbf{u}_i$ at node $v_i$. 
The constraints $\mathbf{u}_i = \mathbf{z}_{ij}$ and $\mathbf{u}_j = \mathbf{z}_{ij}$ imply that for all pairs of agents $(v_i,v_j)$ that form an edge, the feasible set of \eqref{eq:separated-problem1} is such that $\mathbf{u}_i =\mathbf{u}_j$. 
 If the network is strongly connected, these local consensus constraints imply that feasible solutions must satisfy $\mathbf{u}_i =\mathbf{u}_j$ for all, not necessarily neighboring, pairs of agents $v_i$ and $v_j$.\footnote{In contrast to existing literature on  Distributed ADMM algorithm \cite{ling2013decentralized} \cite{wahlberg2012admm}, where only consensus of local estimates is required, the distributed resource allocation problem we consider in addition to being constrained locally, requires consensus in local estimates.} Normalization of the vectors $\mathbf{u}_i$ for $i=1,\hdots,n$ is to ensure that the local estimates $\mathbf{u}_i$ have the same direction.

To compute the augmented Lagrange of \eqref{eq:separated-problem1}, let the dual variable $\alpha_{ij}$ and $\gamma_{ij}$ be associated with equality constraints $\mathbf{u}_i =\mathbf{z}_{ij}$, and  $\mathbf{u}_j =\mathbf{z}_{ij}, \ \forall \ (v_i,v_j)\in \mathcal{E}$ respectively. In the iterations of the D-ADMM algorithm, the dual variable updates are:
\begin{eqnarray}
\alpha_{ij}(k+1) = \alpha_{ij}(k) + \frac{\rho}{2}(\mathbf{u}_i(k) - \mathbf{u}_j(k)) \ \forall \ j\in N(i) \\
\gamma_{ij}(k+1) = \gamma_{ij}(k) + \frac{\rho}{2}(\mathbf{u}_j(k) - \mathbf{u}_i(k)) \ \forall \ j\in N(i).
\end{eqnarray}
Let $\Gamma_i(k+1) \triangleq \beta_i(k+1), \delta_i(k+1), \mathbf{u}_i(k+1)$. Further, let $\phi_i(k) \triangleq \sum_{j\in N(i)} (\alpha_{ij}(k) + \gamma_{ji}(k)) \ \forall \ v_i\in V$. Given $\mathbf{u}_i(0)\in \mathbb{R}^n$ and $\phi_i(0) = \textbf{0}$, the iterative computations and updates are summarized in Algorithm \ref{admm-algorithm1}.
\alglanguage{pseudocode}
\begin{algorithm}[h]
\small
\caption{Distributed ADMM} 
\label{admm-algorithm1}
\begin{algorithmic}[1]
\State \textbf{Given} initial variables $\beta_i(0), \delta_i(0) \in \mathbb{R}, \ \ \mathbf{u}_i(0),  \phi_i(0) \in \mathbb{R}^n$ for each agent $v_i \in \mathcal{V}$. \\ \textbf{Set} $k=1$
\Repeat
\State For all $v_i\in \mathcal{V}$
\begin{equation}
\phi_i(k+1) = \phi_i(k) + \rho\sum_{j\in N(i)} (\mathbf{u}_i(k) - \mathbf{u}_j(k))
\end{equation}
\begin{equation} 
\begin{aligned}
\Gamma_i(k+1) & = \text{arg}\min_{\beta_i, \delta_i, \mathbf{u}_i} \ f_i(\beta_i) + g_i(\delta_i) + \phi_i^T \mathbf{u}_i \\
& + \rho \sum_{j\in N(i)} \|\mathbf{u}_i - \frac{\mathbf{u}_i(k) + \mathbf{u}_j(k)}{2} \|_2^2 \\
\mbox{subject to } & \frac{\beta_{i}\sum_{j=1}^{n}A_{ij}\mathbf{u}_i^i +  \delta_i \mathbf{u}_i^j}{\mathbf{u}_i^i}\leq 1, \label{eq:separated problem1} \\
&  \prod_{j=1}^n 
\mathbf{u}_i^j = 1, \nonumber \\ 
 & 
 \underline{\delta} \leq \delta_i \leq \overline{\delta} \nonumber \\
& 
\underline{\beta} \leq \beta_i \leq \overline{\beta}. \nonumber \\
\end{aligned}
\end{equation}
\State \textbf{Set} $k = k+1$
\Until $\sum_{i=1}^n \sum_{j\in N(i)} \|\mathbf{u}_i (k) - \mathbf{u}_j(k)\| \leq \eta$, for $\eta$ arbitrarily small. 

\Statex
\end{algorithmic}
  \vspace{-0.4cm}%
\end{algorithm}


\section{Numerical Simulations}
\label{sec:experiments}
In this section, we illustrate performance of the D-ADMM Algorithm \ref{admm-algorithm1} on two strongly connected directed networks, and briefly  discuss its convergence. As a proof of concept, and to show that functional correctness of Algorithm \ref{admm-algorithm1}, we will show that the investment in vaccines and antidotes for all agents in the network converge to the solution obtained in the centralized case. Further, we will show convergence of the of the dual variables $\phi_i$ for all agents $v_i$. 

The D-ADMM algorithm was used to solve Problem \ref{Problem:decentralized Rate Constrained Allocation} on a $8$-node network with the following parameters: epidemic threshold $\tau_c = (1-\overline{\delta})/\rho(A)$; $\overline{\beta} = 2\tau_c, \ \  \underline{\beta} = 0.2 \overline{\beta}$ and $\overline{\delta} = 0.9,  \ \ \underline{\delta} = 0.5$. These parameters were chosen in such a way as to ensure the DFE is unstable in the absence of any investment in vaccines and/or antidotes, and stable otherwise. We normalize the investment in vaccines and antidotes using the following quasi-convex functions
\begin{equation}\label{eq:Quasiconvex Limit-2}
f_{i}\left(\beta_{i}\right)=\frac{\beta_{i}^{-1}-\bar{\beta}_{i}^{-1}}{\underline{\beta}_{i}^{-1}-\bar{\beta}_{i}^{-1}},\: \quad g_{i}\left(\delta_{i}\right)=\frac{\left(1-\delta_{i}\right)^{-1}-\left(1-\underline{\delta}_{i}\right)^{-1}}{\left(1-\overline{\delta}_{i}\right)^{-1}-\left(1-\underline{\delta}_{i}\right)^{-1}}.
\end{equation}
The functions are such that for $\beta_i = \overline{\beta}_i$,  $f_i(\overline{\beta}_i) = 0$ and for $\beta_i = \underline{\beta}_i$,  $f_i(\underline{\beta}_i) = 1$. Similarly, for $\delta_i = \overline{\delta}_i$,  $g_i(\overline{\delta}_i) = 1$ and for $\delta_i = \underline{\delta}_i$,  $g_i(\underline{\delta}_i) = 0$. The Lagrangian penalty parameter $\rho$ was chosen to be $4$.

We illustrate the distributed solution on an $8$-node directed network, where the probability of a directed edge between two points is $0.32$. The values of $\underline{\delta}= 0.025$ and $\overline{\delta} = 0.750$ and $\underline{\beta} = 0.1142$ and $\overline{\beta}= 0.4393$. Since a feasibility constraint is $\rho(BA - D) < 1$, the upper and lower bounds of $\beta_i$ and $\delta_i$ are such that\footnote{Specifically, the upper and lower bounds of $\beta_i$ and $\delta_i$ were chosen as follows and based on the spectrum of $A$. We have $\overline{\delta} = 0.8$, \ $\underline{\delta} = 3.9\times \frac{2}{10}$, $\tau_c = (2/10)/\lambda_1(A)$, $\overline{\beta} = 4 \tau_c$, $\underline{\beta} = 30\% \overline{\beta}$, all chosen in a way that ensures infeasibility of a solution when all agents are assigned the maximum or minimum possible infection or recovery rates.} when $\beta_i = \underline{\beta}$ and $\delta_i = \underline{\delta}$, across all agents in the network, $\rho(BA -D) = 0.4600$. Further, when $\beta_i = \overline{\beta}$ and $\delta_i = \overline{\delta}$, $\rho(BA -D) = 1.6$. And when $\beta_i = \underline{\beta}$ and $\delta_i = \overline{\delta}$, $\rho(BA -D) = 1.04$. Finally, when $\beta_i = \overline{\beta}$ and $\delta_i = \underline{\delta}$, $\rho(BA -D) = 1.02$. These bounds ensure that all agents in the network are not easily allocated resources to yield the minimum possible infection rate or the maximum possible recovery rates in the network, since the epidemic control criterion will be violated.
%
With the above parameters, the optimal solution to the centralized problem using the cost functions in \eqref{eq:Quasiconvex Limit-2} was $1.9731$. As can be seen in Figure \ref{fig:cost}, we observe convergence of the total investment in the distributed solution to that obtained in the centralized solution. Figure \ref{fig:conv_of_u} illustrates the convergence of errors in the local estimates $\mathbf{u}_i$ across all agents in the network. It shows that the local estimates $\mathbf{u}_i$ of the global variable at node each node $v_i$ converges to those of their neighbors and the agents reach consensus in their estimates. As illustrated in Figure \ref{fig:dual_phi}, the dual variables $\phi_i$ for all nodes converge. 

\begin{figure*}
\minipage{0.33\textwidth}
  \includegraphics[width=\linewidth]{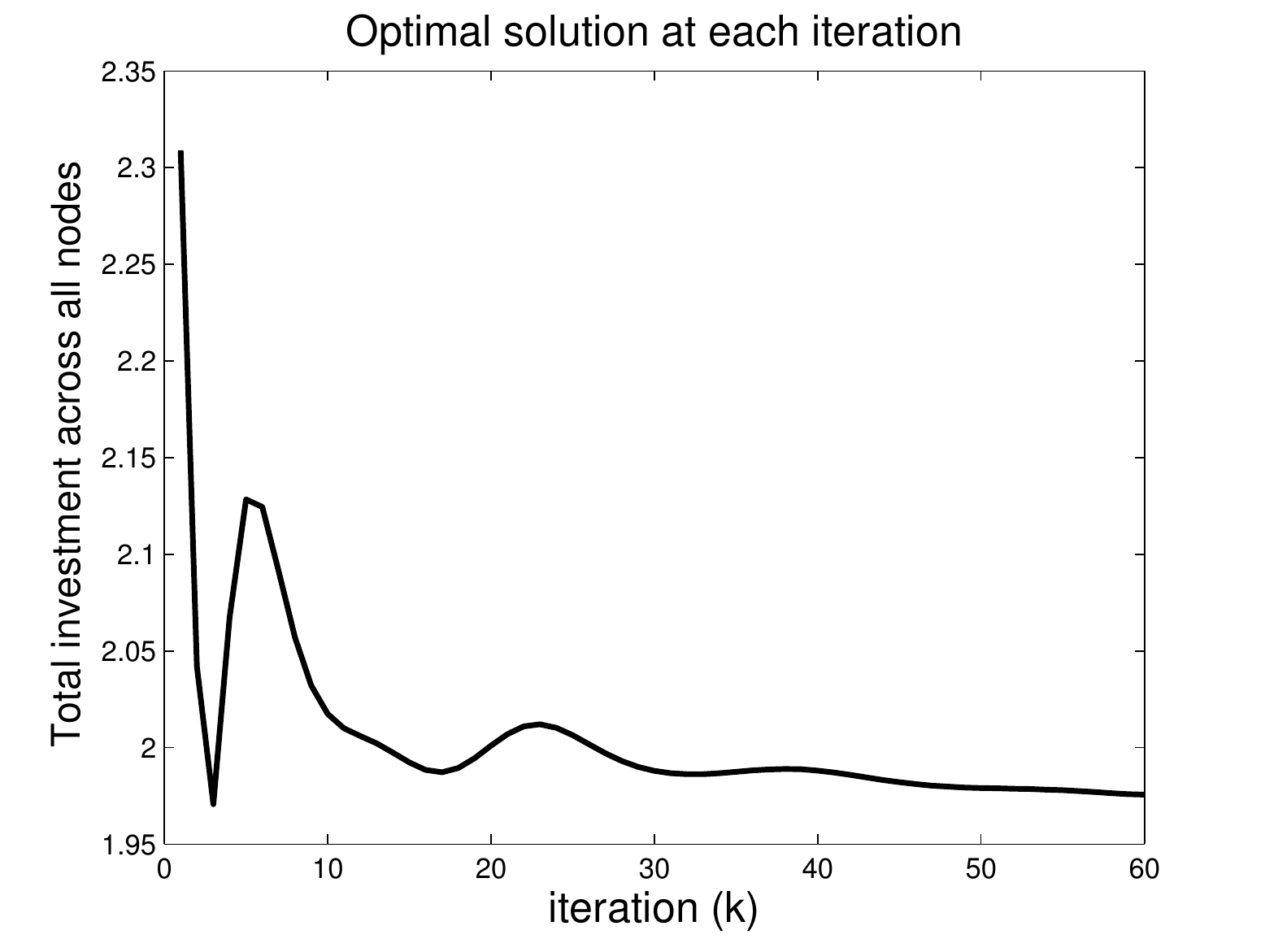}
  \caption{Convergence of Optimal solution}\label{fig:cost}
\endminipage\hfill
\minipage{0.33\textwidth}
  \includegraphics[width=\linewidth]{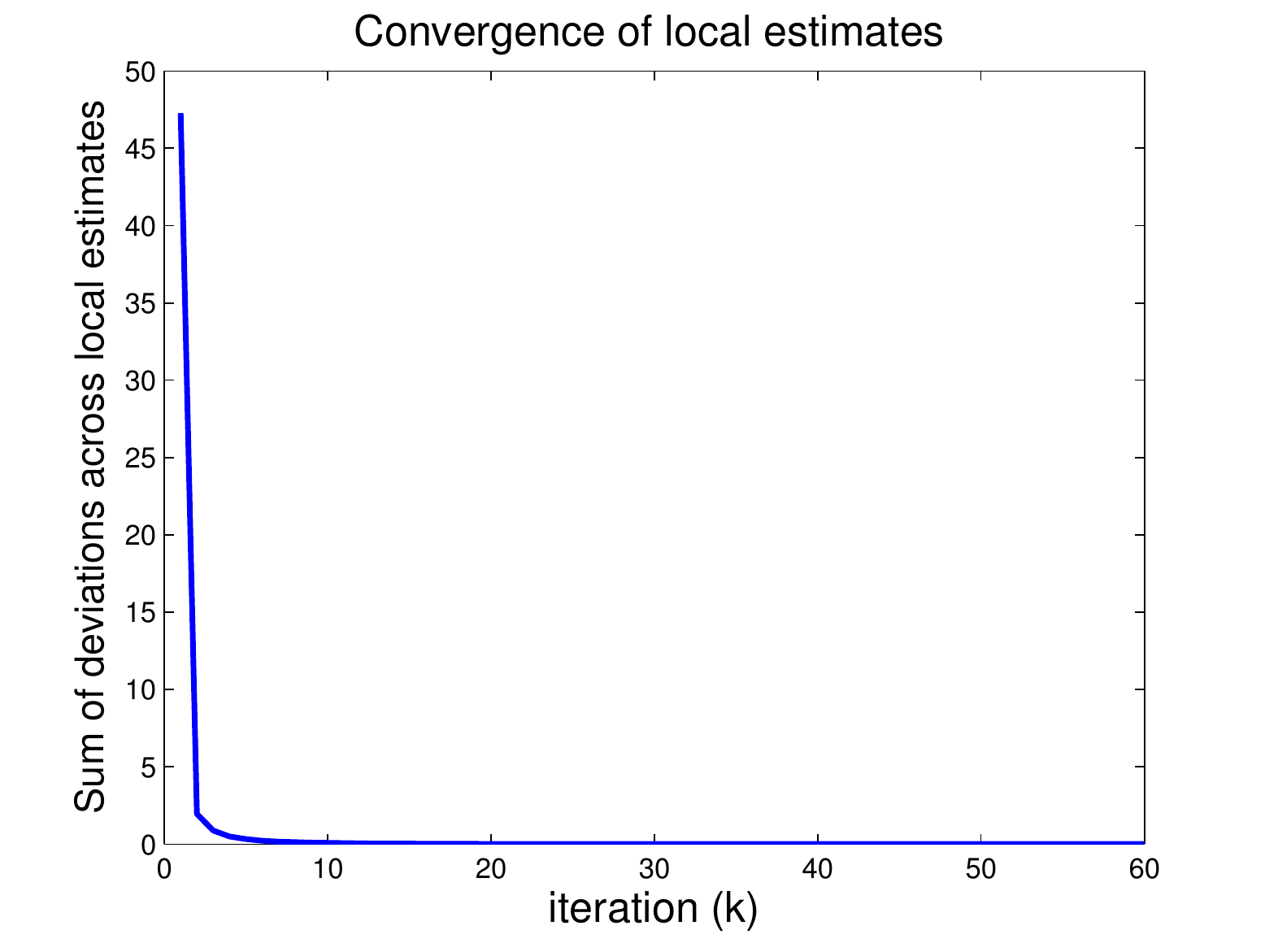}
  \caption{Error convergence of global estimates}\label{fig:conv_of_u}
\endminipage\hfill
\minipage{0.33\textwidth}%
  \includegraphics[width=\linewidth]{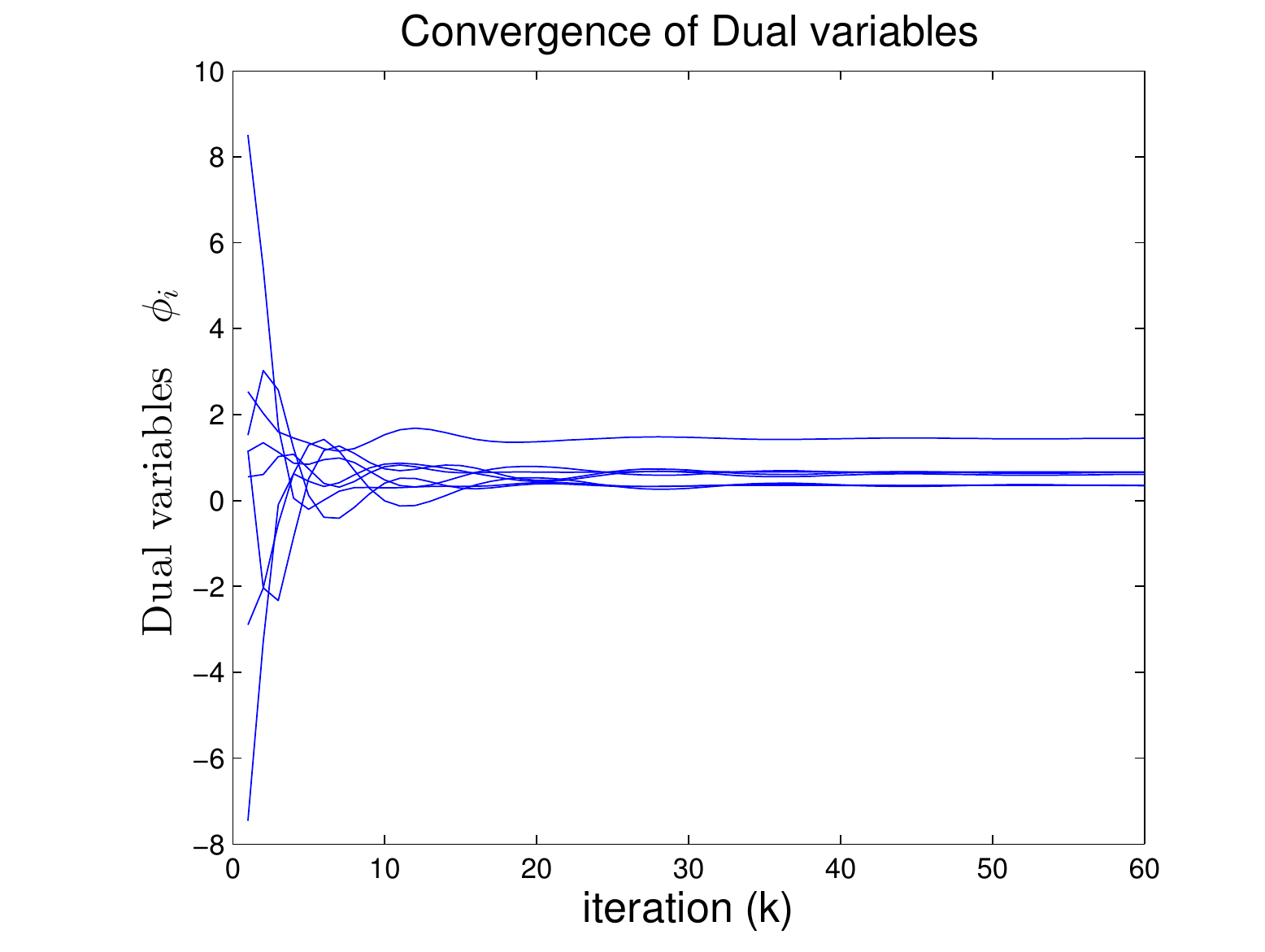}
  \caption{Convergence of dual variable}\label{fig:dual_phi}
\endminipage
\end{figure*}

We illustrate the solution on a fairly larger strongly connected network comprising $20$ nodes. As was done earlier, the values $\underline{\delta}= 0.025$, $\overline{\delta} = 0.750$, $\underline{\beta} = 0.0641$ and $\overline{\beta}= 0.2464$ are chosen such that when $\beta_i = \underline{\beta}$ and $\delta_i = \underline{\delta}$, across all agents in the network, $\rho(BA -D) = 0.3500$. Further, when $\beta_i = \overline{\beta}$ and $\delta_i = \overline{\delta}$, $\rho(BA -D) = 2.000$. And when $\beta_i = \underline{\beta}$ and $\delta_i = \overline{\delta}$, $\rho(BA -D) = 1.0750$. Finally, when $\beta_i = \overline{\beta}$ and $\delta_i = \underline{\delta}$, $\rho(BA -D) = 1.2750$. 

\begin{figure*}
\minipage{0.33\textwidth}
  \includegraphics[width=\linewidth]{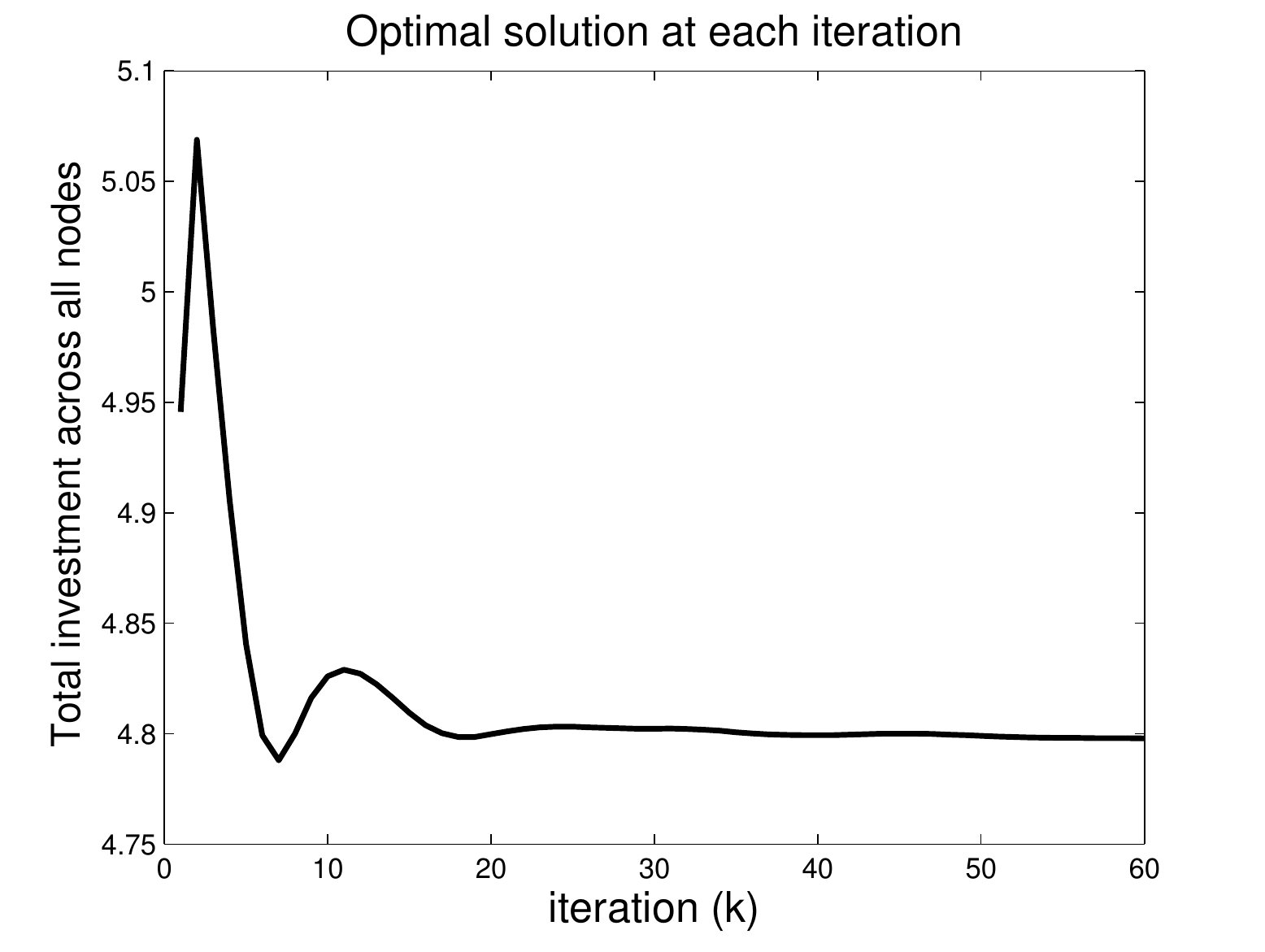}
  \caption{Convergence of Optimal solution}\label{fig:cost2}
\endminipage\hfill
\minipage{0.33\textwidth}
  \includegraphics[width=\linewidth]{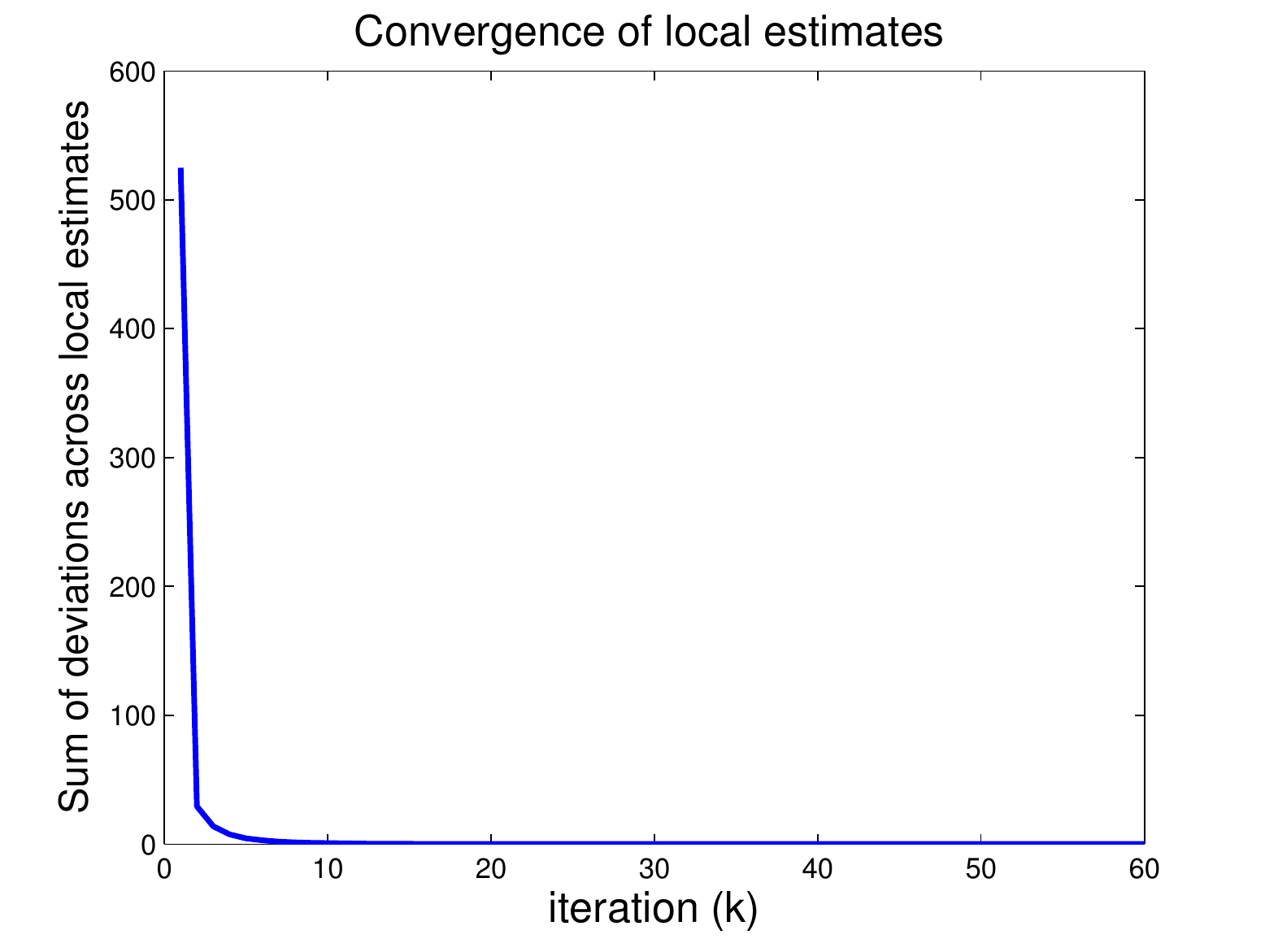}
  \caption{Error convergence of global estimates}\label{fig:conv_of_u2}
\endminipage\hfill
\minipage{0.33\textwidth}%
  \includegraphics[width=\linewidth]{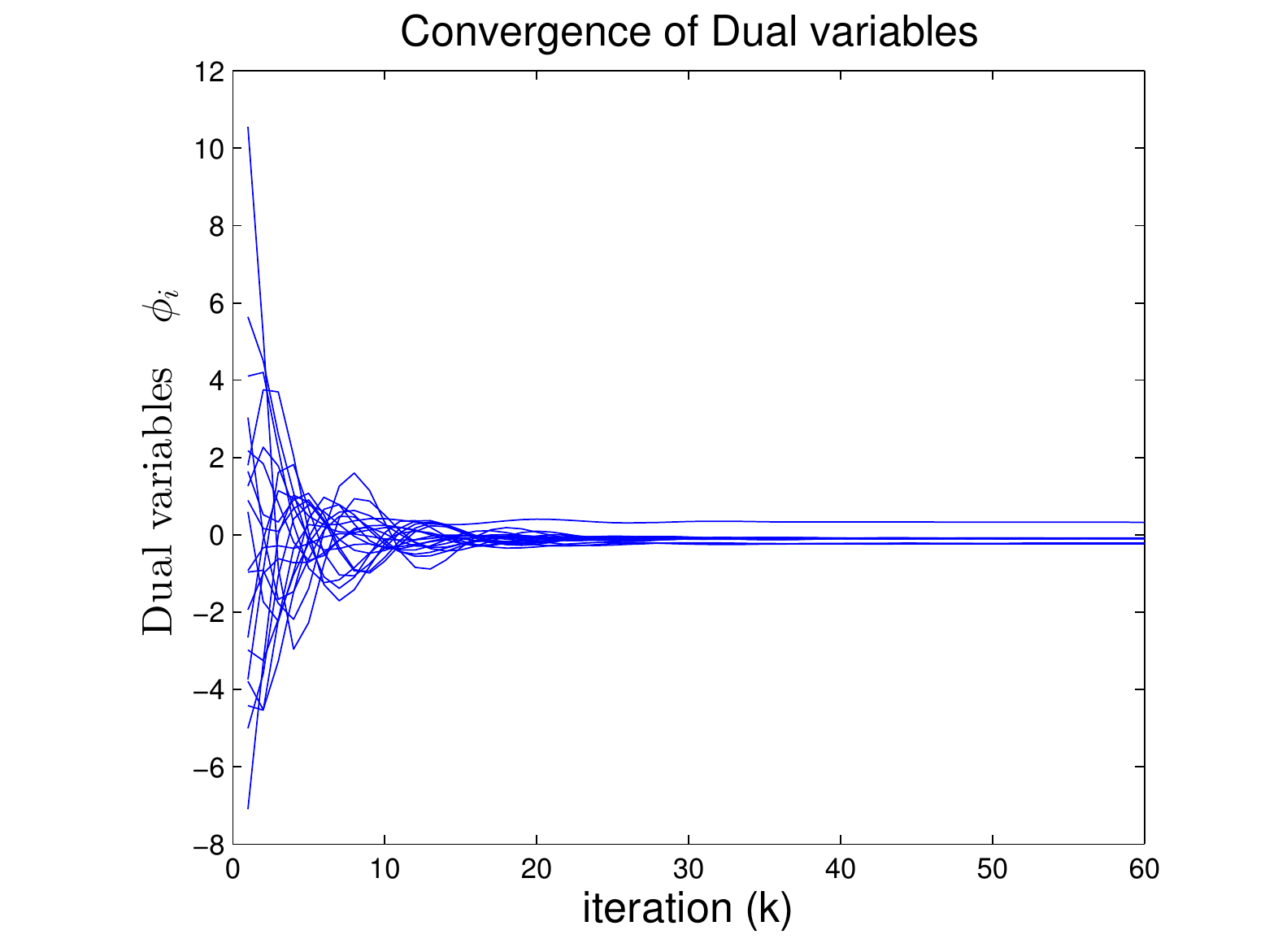}
  \caption{Convergence of dual variable}\label{fig:dual_phi2}
\endminipage
\end{figure*}
The total investment obtained from the centralized solution for the $20$-node network was $4.8868$. In Figure \ref{fig:cost2}, we find that the total investment obtained from the distributed solution converges to that obtained in the centralized solution. We find that the local estimates $\mathbf{u}_i$ of the nodes reach consensus as illustrated in Figure \ref{fig:conv_of_u2}, where we show convergence of the errors in estimates to zero.
Finally, the dual variables, do indeed, converge as illustrated in Figure \ref{fig:dual_phi2}.

\subsection{On convergence of the D-ADMM Algorithm \ref{admm-algorithm1}} 
The ADMM algorithm applied to distributed optimization problems have typically considered distributed problems where the only constraint is consensus in the local estimates of the agents; for example \cite{boyd2011distributed,ling2013decentralized, chang2014multi}. With just a consensus constraint and smoooth, differentiable convex cost, explicit computation of the local decision variables is possible. This allows for analytical expression of the optimal iterates, which enables convergence rate analysis of the algorithm. 

The problem considered in this paper, however, is a constrained optimization problem where, in addition to the consensus constraint, each agent also has three local constraints to satisfy to achieve a feasible solution at each iteration of the algorithm as specified in \eqref{eq:separated-problem1}. This informed the use of numerical solvers for computing the optimal local variables at each node, as done in Algorithm \ref{admm-algorithm1}. 

\begin{remark}\label{rem:lse-convex}
It is known that ADMM algorithm converges when applied to convex problems \cite{boyd2011distributed,ling2013decentralized}. The convex characterization of our resource allocation problem via GP presented in Section \ref{sec:convex_character} guarantees that the D-ADMM solution in Algorithm \ref{admm-algorithm1} converges. The use of numerical solvers in computing the local optimal solution at each node (in line $4$ of Algorithm \ref{admm-algorithm1}) makes it difficult carrying out a convergence rate analysis. 
\end{remark}

\section{Summary}
\label{sec:summary}
In this paper, we proposed a fully distributed solution to the problem of optimally allocating vaccine and antidote investment to control an epidemic outbreak in a networked population at a desired rate. The proposed solution was a D-ADMM algorithm, enabling each node to locally compute it's optimum investment in vaccine and antidotes needed to globally contain the spread of an outbreak, via local exchange of information with its neighbors. In contrast to previous literature, our problem is a constrained optimization problem associated with a directed network comprising non-identical agents.  Since numerical solvers are used to solve convex subproblems at each node, a convergence rate analysis of the D-ADMM algorithm is not presented. However, it is known that the ADMM algorithm converges for convex problems \cite{boyd2011distributed}; further, illustration of our results via simulations in Section \ref{sec:experiments} show that that the D-ADMM algorithm converges. The proposed distributed solution to the vaccine and antidote allocation problem for epidemic control presents a framework to contain outbreak in the absence of a central social planner.


\end{document}